\documentclass[letterpaper,11pt]{article}
\usepackage[latin1]{inputenc}
\usepackage[T1]{fontenc}
\usepackage[in]{fullpage}
\usepackage{amsmath,amssymb}
\usepackage{float} \usepackage{graphicx}
\usepackage{tikz}
\usetikzlibrary{calc}
\usepgflibrary{arrows}
\usepackage{amsthm}
\usepackage{mdframed} \usepackage{bbm}      \usepackage{hyperref}
\usepackage{pdfpages} \usepackage{wasysym}  \usepackage{lmodern}  \usepackage{microtype}\usepackage{verbatim}
\usepackage{wrapfig}
\usepackage{multirow} \usepackage[ruled,linesnumbered,lined,boxed,commentsnumbered]{algorithm2e}

\newcommand{\set}[1]{\left\{#1\right\}}

\newcommand{\sizeof}[1]{\left\lvert#1\right\rvert}

\newcommand{\OO}{\mathcal{O}}

\theoremstyle{plain}
\newtheorem{theorem}{Theorem}[]
\newtheorem{observation}[theorem]{Observation}
\newtheorem{lemma}[theorem]{Lemma}
\newtheorem{corollary}[theorem]{Corollary}

\theoremstyle{definition}
\newtheorem{definition}[theorem]{Definition}

\usepackage{authblk} \title{One-Way Trail Orientations}
\author{Anders Aamand}
\author{Niklas Hjuler \thanks{This work is supported by the Innovation Fund Denmark
through the DABAI project.}}
\author{Jacob Holm\thanks{This research is supported by Mikkel Thorup's Advanced Grant DFF-0602-02499B from the Danish Council for Independent Research under the Sapere Aude research career programme.}}
\author{\\Eva Rotenberg}
\affil{University of Copenhagen}
\date{}

\begin{document}
\maketitle

\begin{abstract}
Given a graph, does there exist an orientation of the edges such that the resulting directed graph is strongly connected? 
  Robbins' theorem [Robbins, Am. Math. Monthly, 1939] states that  such an orientation exists if and only if the graph is $2$-edge connected. A natural extension of this problem is the following: Suppose that the edges of the graph is partitioned into trails. Can we orient the trails such that the resulting directed graph is strongly connected? 
  
   We show that $2$-edge connectivity is again a sufficient condition and we provide a linear time algorithm for finding such an orientation, which is both optimal and the first polynomial time algorithm for deciding this problem.

  The generalised Robbins' theorem [Boesch, Am. Math. Monthly, 1980] for mixed multigraphs states that the undirected edges of a mixed multigraph can be oriented making the resulting directed graph strongly connected exactly when the mixed graph is connected and the underlying graph is bridgeless. We show that as long as all cuts have at least $2$ undirected edges or directed edges both ways, then there exists an orientation making the resulting directed graph strongly connected. This provides the first polynomial time algorithm for this problem and a very simple polynomial time algorithm to the previous problem.
\end{abstract}

\newpage

\section{Introduction and motivation}

Suppose that the mayor of a small town decides to make all the streets one-way in such a way that it is possible to get from any place to any other place without violating the orientations of the streets\footnote{The motivation for doing so is that the streets of the town are very narrow and thus it is a great hassle when two cars unexpectedly meet.}. If initially all the streets are two-way then Robbins' theorem \cite{Robbins39} asserts that this can be done exactly when the corresponding graph is $2$-edge connected. If, on the other hand some of the streets were already one-way in the beginning then the generalised Robbins' theorem \cite{Boesch80} states that it can be done exactly when the corresponding graph is strongly connected and the underlying graph is $2$-edge connected. 
 
However, the proofs of both of these results assume that every street of the city corresponds to exactly one edge in the graph. This assumption hardly holds in any city in the world and therefore a much more natural assumption is that every street corresponds to a trail in the graph and  that the edges of each trail must be oriented consistently\footnote{This version of the problem was given to us through personal communication with Professor Robert E. Tarjan}.
 
 In this paper we prove that Robbins' Theorem continues to hold even when the set of edges is partitioned into trails. In other words a necessary and sufficient condition for an orientation to exist is that the graph is $2$-edge connected. We also provide a linear time algorithm for finding such an orientation. 

Finally we will consider the generalised Robbins' theorem in this new setting i.e. we allow some edges to be oriented initially and suppose that the remaining edges are partitioned into trails. We will show that if any cut $(V_1,V_2)$ in the graph has either at least $2$ undirected edges going between $V_1$ to $V_2$  or a directed edge in each direction then it is possible to orient the trails making the resulting graph strongly connected. Although this condition is not necessary it does give a simple algorithm for deciding the problem. Indeed, the only cuts containing an undirected edge which we allow are the ones where this edge (and hence its trail) is \emph{forced} in one direction. Hence for deciding the problem we can start by orienting all the forced trails until there are no more forced trails. Then the trails can be oriented making the graph strongly connected exactly if the resulting graph satisfies our condition. 

Note that when some edges are initially oriented the answer to the problem depends on the trail decomposition which is not the case for the other results. That the condition from the generalised Robbins' theorem is not sufficient can be seen from figure \ref{fig:counter example}.

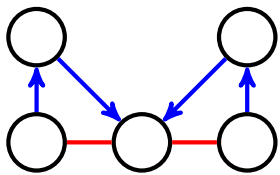
\begin{figure}[h!]
\centering
\begin{tikzpicture}[x=1.4cm, y=1.4cm]
  \begin{scope}[
      vertex/.style={
        ultra thick,
        circle,
        draw,
        fill=white,
        minimum size=0.75cm,
      },
      edge/.style={
        ultra thick,
      },
      undirected edge/.style={
        edge,
      },
      directed edge/.style={
        edge,
        ->,>=stealth',
      },
    ]
    \node[vertex] (a) at (0,0) {};
    \node[vertex] (b) at (1,0) {};
    \node[vertex] (c) at (2,0) {};
    \node[vertex] (d) at (0,1) {};
    \node[vertex] (e) at (2,1) {};

    \draw[undirected edge,red] (a) -- (b);
    \draw[undirected edge,red] (b) -- (c);
    \draw[directed edge,blue] (a) -- (d);
    \draw[directed edge,blue] (d) -- (b);
    \draw[directed edge,blue] (c) -- (e);
    \draw[directed edge,blue] (e) -- (b);
  \end{scope}
\end{tikzpicture}
\caption{The graph is strongly connected and the underlying graph is $2$-edge connected, but no matter the orientation of the red trail, the graph will lose its strong connectivity}
\label{fig:counter example}
\end{figure}
\paragraph{Earlier methods}
Several methods have already been applied for solving orientation problems in graphs where the goal is to make the resulting graph strongly connected. 

One approach used by Robbins \cite{Robbins39} is to use that a $2$-edge connected graph has an \emph{ear-decomposition}. An ear decomposition of a graph is a partition of the set of edges into a cycle $C$ and paths $P_1,\dots,P_t$ such that $P_i$ has its two endpoints but none of its internal vertices on $C\cup\left(\bigcup_{j=1}^{i-1}P_j\right)$. Assuming the existence of an ear decomposition of $2$-edge connected graphs it is easy to prove Robbins' theorem. Indeed, it is easy to see by induction that any consistent orientations of the paths and the cycle give a strongly connected graph.

A second approach introduced by Tarjan \cite{Hopcroft:1973:AEA:362248.362272} gives another simple proof of Robbins' theorem. One can make a DFS tree in the graph rooted at a vertex $v$ and orient all edges in the DFS tree away from $v$. The remaining edges are oriented towards $v$ and if the graph is $2$-edge connected it is easily verified that this gives a strong orientation.

A similar approach was used by Chung et al. \cite{Chung:1985} in the context of the generalized Robbins theorem for mixed multigraphs.

The above methods not only prove Robbins' theorem, they also provide linear time algorithms for finding strong orientations of undirected or mixed multigraphs.

However, none of the above methods have proven fruitful in our case. In case of the ear decomposition one needs a such which is somehow compatible with the partitioning into trails and this seems hard to guarantee. The original proof by Roberts is essentially similar to using the ear decomposition.  Similar problems appear when trying a DFS-approach. Neither does the proof by Boesch \cite{Boesch80} of Robbins' theorem for mixed multigraphs generalise to prove our result. Most importantly the corresponding theorem is no longer true for trail orientations as is shown by the example above.

Since the classical linear time algorithms rely on ear-decompositions and DFS searches, and since these approaches do not immediately work for trail partitions, our linear time algorithm will be a completely new approach to solving orientation problems.

\section{Preliminaries}

Let us briefly review the concepts from graph theory that we will need.
Recall that a \emph{walk} in a graph is an alternating sequence of vertices and edges $v_0,e_1,v_1,e_2,\dots,v_k$, such that for $1\leq i \leq k$ the edge $e_i$ has $v_{i-1}$ and $v_i$ as its two endpoints. In a directed or mixed graph the ordering of the endpoints of each edge in the sequence must be consistent with the direction of the edge in case it is oriented. A \emph{trail} is a walk without repeated edges. A \emph{path} is a trail without repeated vertices (except possibly $v_0=v_k$). Finally a \emph{cycle} is a path for which $v_0=v_k$

Next, recall that a mixed multigraph $G=(V,E)$ is called \emph{strongly connected} if for any vertices $u,v\in V$ there exists a walk from $u$ to $v$. In case that the graph contains no directed edges this is equivalent to saying that it consists of exactly one connected component. 

We also recall the definition of $k$-edge connectivity. A graph $G=(V,E)$ is said to be \emph{$k$-edge connected} if and only if  $G'=(V,E-X)$ is connected for all $X \subseteq E$ where $|X|<k$. A trivially equivalent condition is that any edge-cut $(V_1,V_2)$ in the graph has at least $k$ edges going between $V_1$ and $V_2$. 

Finally, if $G=(V,E)$ is a mixed multigraph and $A\subseteq V$ we define $G/A$ to be that graph obtained by contracting $A$ to a single vertex and $G[A]$ to be the subgraph of $G$ induced by $A$. The following simple observation will be used repeatedly in this paper.
\begin{observation}
  If $G=(V,E)$ is $k$-edge connected and $A\subseteq V$ then $G/A$ is $k$-edge connected. Also if $G$ is a strongly connected mixed multigraph then $G/A$ is too.
\end{observation}
The structure of this paper is as follows. In section \ref{sec:rob} we prove our generalisation of Robbins' theorem for undirected graphs partitioned into trails. In section \ref{sec:mix} we study what happens in the case of mixed graphs. Finally in section \ref{sec:lin} we provide our linear time algorithm for trail orientation in an undirected graph.

\section{Robbins Theorem Revisited}\label{sec:rob}
We are now ready to state our generalisation of Robbins' theorem. 
\begin{theorem}\label{thm:exist}
Let $G=(V,E)$ be a multigraph with $E$ partitioned into trails. An orientation of each trail such that the resulting directed graph is strongly connected exists if and only if $G$ is $2$-edge connected.
\end{theorem}
\begin{proof}
If  $G$ is not $2$-edge connected, such an orientation obviously doesn't exist so we need to prove the converse.  Suppose therefore that $G$ is $2$-edge connected.

Our proof is by induction on the number of edges in $G$.  If there are no edges, the graph is a single vertex, and the statement is obviously true.  Assume now the statement holds for all graphs with strictly fewer edges than $G$.  Pick an arbitrary edge $e$ that is at the end of its corresponding trail.

If $G-e$ is $2$-edge connected, then by the induction hypothesis there is a strong orientation of $G-e$ that respects the trails of $G$.  Such an orientation clearly extends to the required orientation of $G$.

\begin{figure}[h!]
\centering
\begin{tikzpicture}[x=0.5cm,y=0.5cm,scale=0.75]
  \begin{scope}[
      every path/.style={
              },
      every node/.style={
                text=black,
        inner sep=1pt,
      },
            vertex set/.style={
        dashed,
      },
            vertex/.style={
        draw,
        circle,
        fill=white,
        minimum size=2mm,
        inner sep=0pt,
        outer sep=0pt,
      },
            edge/.style={blue,thick},
      undirected edge/.style={edge},
      directed edge/.style={edge,->,>=stealth'},
    ]
    \begin{scope}
      \draw[vertex set] (0,0) node (V1) {$V_1$} ellipse (3 and 4);
      \draw[vertex set] (12,0) node (V2) {$V_2$} ellipse (3 and 4);
      \node[vertex,label={60:$u_1$}] (e1) at ($(V1)+(2,3)$) {};
      \node[vertex,label={-60:$w_1$}] (b1) at ($(V1)+(2,-3)$) {};
      \node[vertex,label={120:$u_2$}] (e2) at ($(V2)+(-2,3)$) {};
      \node[vertex,label={-120:$w_2$}] (b2) at ($(V2)+(-2,-3)$) {};
      \draw[undirected edge] (e1) -- (e2) node[midway,label={above:$e$}] {};
      \draw[undirected edge] (b1) -- (b2) node[midway,label={below:$b$}] {};
    \end{scope}

    \begin{scope}[
                shift={(24,0)},
      ]
      \draw[vertex set] (0,0) node (V1) {$V_1$} ellipse (3 and 4);
      \draw[vertex set] (12,0) node (V2) {$V_2$} ellipse (3 and 4);
      \node[vertex,label={60:$u_1$}] (e1) at ($(V1)+(2,3)$) {};
      \node[vertex,label={-60:$w_1$}] (b1) at ($(V1)+(2,-3)$) {};
      \node[vertex,label={120:$u_2$}] (e2) at ($(V2)+(-2,3)$) {};
      \node[vertex,label={-120:$w_2$}] (b2) at ($(V2)+(-2,-3)$) {};
      \draw[directed edge] (e1)
      .. controls ($(V1)+(6.25,1.25)$) and ($(V1)+(6.35,-1.25)$)
      .. (b1);
      \draw[directed edge] (b2)
      .. controls ($(V2)+(-6.25,-1.25)$) and ($(V2)+(-6.25,1.25)$)
      .. (e2);
    \end{scope}

  \end{scope}
\end{tikzpicture}
\caption{A two edge cut and the two graphs $G_1$ and $G_2$. }
\label{fig:Two edge cut}
\end{figure}
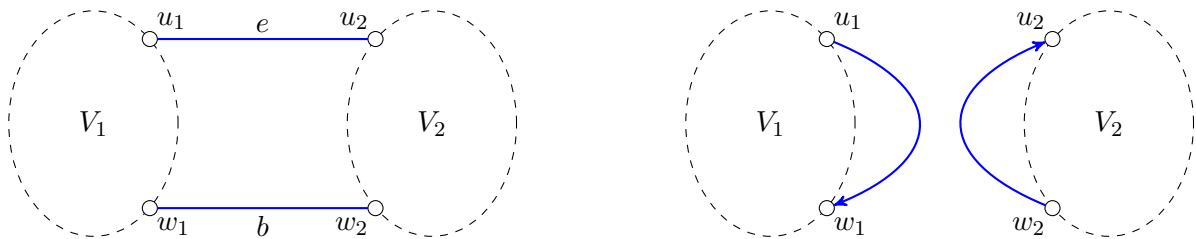

If $G-e$ is not $2$-edge connected, there exists a bridge $b$ in $G-e$ (see figure \ref{fig:Two edge cut}).  Let $V_1$,~$V_2$ be the two connected components of $G-\set{e,b}$, and let $e=(u_1,u_2)$ and $b=(w_1,w_2)$ such that for $i\in\set{1,2}$, $u_i,w_i\in V_i$ (note that we don't necessarily have that $u_i$ and $w_i$ are distinct for $i\in \{1,2\}$).  Now for $i\in\set{1,2}$ construct the graph $G_i = G[V_i]\cup\set{(u_i,w_i)}$, and define the trails in $G_i$ to be the trails of $G$ that are completely contained in $G_i$, together with a single trail combined from the (possibly empty) partial trail of $e$ contained in $G_i$ and ending at $e_i$, followed by the edge $(u_i,w_i)$, followed by the (possibly empty) partial trail of $b$ contained in $G_i$ starting at $b_i$.  Both $G_1$ and $G_2$ are $2$-edge connected since they can each be obtained as a contraction of $G$. Furthermore, they each have strictly fewer edges than $G$, so inductively each has a strong orientation that respects the given trails.  Further, we can assume that the orientations are such that the new edges are oriented as $(u_1,w_1)$ and $(w_2,u_2)$ by flipping the orientation of all edges in either graph if necessary.  We claim that this orientation, together with $e$ oriented as $(u_1,u_2)$ and $b$ oriented as $(w_2,w_1)$, is the required orientation of $G$.  To see this first note that (by construction) this orientation respects the trails.  
Secondly suppose $v_1\in V_1$ and $v_2\in V_2$ are arbitrary. Since $G_1$ is strongly connected $G[V_1]$ contains a directed path from $v_1$ to $u_1$. Similarly, $G[V_2]$ contains a directed path from $u_2$ to $v_2$. Thus $G$ contains a directed path from $v_1$ to $v_2$. A similar argument gives a directed path from $v_2$ to $v_1$ and since $v_1$ and $v_2$ were arbitrary this proves that $G$ is strongly connected and our induction is complete.

\end{proof}

The construction in the proof can be interpreted as a naive algorithm for finding the required orientation when it exists.
\begin{corollary}
  The one-way trail orientation problem on a graph with $n$ vertices and $m$ edges can be solved in $\OO(n+m\cdot f(m,n))$ time, where $f(m,n)$ is the time per operation for fully dynamic bridge finding (a.k.a. $2$-edge connectivity).
\end{corollary}
At the time of this writing\footnote{Separate paper submitted to SODA'18 by Holm, Rotenberg and Thorup.}, this is $\OO(n+m(\log n\log\log n)^2)$.  In Section~\ref{sec:lin} we will show a less naive algorithm that runs in linear time.

\section{Extension to Mixed graphs}\label{sec:mix}
Now we will extend our result to the case of mixed graphs. We are going to prove the following.
\begin{theorem}\label{thm:mixed}
  Let $G=(V,E)$ be a strongly connected mixed multigraph. Then $G-e$ is strongly connected for all undirected $e\in E$ if and only if for any partition $\mathcal{P}$ of the undirected edges of $G$ into trails, and any $T\in\mathcal{P}$, any orientation of $T$ can be extended to a strong trail orientation of $(G,\mathcal{P})$.
\end{theorem}
Suppose $G=(V,E)$ is as in the theorem. We will say that $e\in E$ is \emph{forced} if it is undirected and satisfies that $G-e$ is not strongly connected\footnote{This terminology is natural since it is equivalent to saying that there exists a cut $(V_1,V_2)$ in $G$ such that $e$ is the only undirected edge in this cut and such that all the directed edges go from $V_1$ to $V_2$. If one wants an orientation of the trails making the graph strongly connected we are clearly forced to orient $e$ from $V_2$ to $V_1$.}.
Note that this is a proper extension of Theorem \ref{thm:exist} since if $G$ is undirected and $2$-edge connected then no $e\in E$ is forced.

For proving the result we'll need the following lemma.
\begin{lemma}\label{lem:GcontractAB}
  Let $G$ be a directed graph, and let $(A,B)$ be a cut with exactly one edge crossing from $A$ to $B$ and at least one edge crossing from $B$ to $A$.  Then $G$ is strongly connected if and only if $G/A$ and $G/B$ are.
\end{lemma}
\begin{proof}
  Strong connectivity is preserved by contractions, so if $G$ is strongly connected then $G/A$ and $G/B$ both are.  For the other direction, let $(a_1,b_1)$ be the edge going from $A$ to $B$, and let $(b_2,a_2)$ be any edge from $B$ to $A$.  Since $G/A$ is strongly connected and $(a_1,b_1)$ is the only edge from $A$ to $B$, $G/A$ contains a path from $b_1$ to $b_2$ that stays in $B$.  Since this holds for any edge going from $B$ to $A$, and since $G/B$ is strongly connected, $A$ is strongly connected in $G$.  By a symmetric argument, $B$ is also strongly connected in $G$ and since the cut has edges in both directions, $G$ must be strongly connected.
\end{proof}
Now we provide the proof of Theorem~\ref{thm:mixed}. 
\begin{proof}[Proof of theorem~\ref{thm:mixed}]
  If $G-e$ is not strongly connected, the trail $T$ containing $e$ can at most be directed one way since $e$ is forced, so there is an orientation of $T$ which not extend to a strong trail orientation of $(G,\mathcal{P})$. To prove the converse suppose $G-e$ is strongly connected for all undirected $e\in E$.

  The proof is by induction on $\sizeof{\mathcal{P}}$.  If $\sizeof{\mathcal{P}}\leq 1$ the result is trivial.  So suppose $\sizeof{\mathcal{P}}>1$ and that the theorem holds for all $(G',\mathcal{P}')$ with $\sizeof{\mathcal{P}'}<\sizeof{\mathcal{P}}$.

  Consider a trail $T\in\mathcal{P}$.  Suppose there is no cut $(A,B)$ that $T$ crosses exactly once, which has exactly one other undirected edge crossing it, and has every directed edge crossing it going from $A$ to $B$.  Then regardless of the orientation of $T$, the resulting graph $G'$ has no undirected edge $e$ such that $G-e$ is not strongly connected.  Thus, by induction $(G',\mathcal{P}\setminus\set{T})$ has a strong trail orientation, which is also a strong trail orientation of $(G,\mathcal{P})$, as desired.
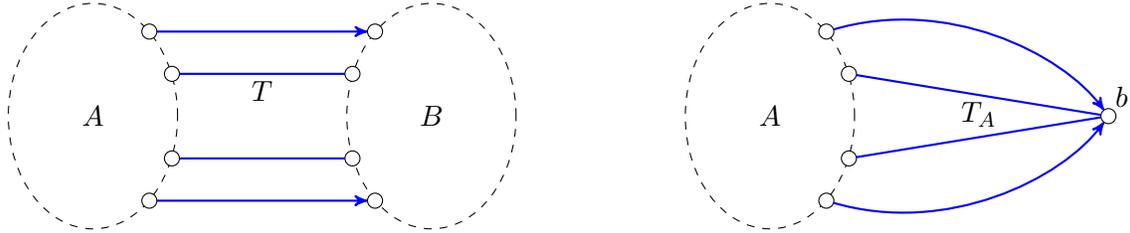
\begin{figure}[h!]
\centering
\begin{tikzpicture}[x=0.5cm,y=0.5cm,scale=0.75]
  \begin{scope}[
      every path/.style={
              },
      every node/.style={
                text=black,
        inner sep=1pt,
      },
            vertex set/.style={
        dashed,
      },
            vertex/.style={
        draw,
        circle,
        fill=white,
        minimum size=2mm,
        inner sep=0pt,
        outer sep=0pt,
      },
            edge/.style={blue,thick},
      undirected edge/.style={edge},
      directed edge/.style={edge,->,>=stealth'},
    ]
    \begin{scope}
      \draw[vertex set] (0,0) node (V1) {$A$} ellipse (3 and 4);
      \draw[vertex set] (12,0) node (V2) {$B$} ellipse (3 and 4);
      \node[vertex,label={60:$ $}] (e1) at ($(V1)+(2,3)$) {};
      \node[vertex,label={-60:$ $}] (b1) at ($(V1)+(2,-3)$) {};
      \node[vertex,label={120:$ $}] (e2) at ($(V2)+(-2,3)$) {};
      \node[vertex,label={-120:$ $}] (b2) at ($(V2)+(-2,-3)$) {};
      \node[vertex,label={80:$ $}] (d1) at ($(V1)+(2.8,1.5)$) {};
      \node[vertex,label={-80:$ $}] (d2) at ($(V1)+(2.8,-1.5)$) {};
      \node[vertex,label={-80:$ $}] (d'1) at ($(V2)+(-2.8,1.5)$) {};
      \node[vertex,label={80:$ $}] (d'2) at ($(V2)+(-2.8,-1.5)$) {};
      \draw[directed edge] (e1) -- (e2) node[midway,label={above:$ $}] {};
      \draw[directed edge] (b1) -- (b2) node[midway,label={below:$ $}] {};
      \draw[undirected edge] (d1) -- (d'1) node[midway,label={below:$T$}] {};
      \draw[undirected edge] (d2) -- (d'2) node[midway,label={below:$ $}] {};
    \end{scope}

    \begin{scope}[
                shift={(24,0)},
      ]
      \draw[vertex set] (0,0) node (V1) {$A$} ellipse (3 and 4);
      \node[vertex,label={60:$b$}] (B) at ($(V1)+(12,0)$) {};
      \node[vertex,label={60:$ $}] (e1) at ($(V1)+(2,3)$) {};
      \node[vertex,label={-60:$ $}] (b1) at ($(V1)+(2,-3)$) {};
      \node[vertex,label={80:$ $}] (d1) at ($(V1)+(2.8,1.5)$) {};
      \node[vertex,label={-80:$ $}] (d2) at ($(V1)+(2.8,-1.5)$) {};
      \draw[directed edge] (e1) 
    .. controls ($(V1)+(6.25,4.25)$) and ($(V1)+(10.35,2.25)$)
      .. (B);
      \draw[directed edge] (b1) .. controls ($(V1)+(6.25,-4.25)$) and ($(V1)+(10.35,-2.25)$)
      .. (B);
      \draw[undirected edge] (d1) -- (B) node[midway,label={below:$T_A$}] {};
      \draw[undirected edge] (d2) -- (B) node[midway,label={below:$ $}] {};

    \end{scope}

  \end{scope}
\end{tikzpicture}
\caption{A cut with two undirected edges and all directed edges going from $A$ to $B$ followed by a contraction of $B$. }
\label{fig:contract}
\end{figure}

  Now suppose there is such a cut $(A,B)$ (see figure \ref{fig:contract}).  Construct a graph $G/B$ by contracting every vertex in $B$ into a single new vertex $b$.  Let $\mathcal{P}_A$ consist of all trails in $\mathcal{P}$ that are completely contained in $A$, together with a single trail $T_A$ combined from the (possibly empty) fragments of the two trails that crossed the cut, joined at $b$.  Since any cut in $G/B$ corresponds to a cut in $G$, $G/B$ is strongly connected and remains so after deletion of any single undirected edge. By   induction any orientation of $T_A$ in $G/B$ extends to a strong orientation of $(G/B,\mathcal{P}_A)$.
Let $G/A$, $a$, $\mathcal{P_B}$ and $T_B$ be defined symmetrically, then by the same argument any orientation of $T_B$ in $G/A$ extends to a strong orientation of $(G/A,\mathcal{P}_B)$.
Now for any orientation of $T$, we can choose orientations of $T_A$ and $T_B$ that are compatible.  The result follows by Lemma~\ref{lem:GcontractAB}.
\end{proof}

Notice that the partitioning of edges into trails does not matter in the case when no edge is forced. Since any undirected graph has no forced edges if it is $2$-edge connected, the theorem implies that the most naive algorithm: "directing trails that are forced and if none are forced direct an arbitrary trail" works for undirected graphs. In general for mixed graphs algorithm \ref{alg:mix} below can clearly be implemented in polynomial time and does solve the trail orientation problem for mixed graphs.

\begin{algorithm}[h!]\label{alg:mix}
      \KwIn{A mixed multigraph $G$ and a partition $\mathcal{P}$ of the undirected edges of $G$ into trails.}
  \KwOut{True if $(G,\mathcal{P})$ has a strong trail orientation, otherwise false.  $G$ is modified in place, either to have such a strong trail orientation, or to a forced graph that is not strongly connected.}
  \If{$G$ has a bridge or is not strongly connected}{
    \Return{false}
  }
  \While{$\sizeof{\mathcal{P}}>0$}{
    \eIf{for some undirected edge $e$, $G-e$ is not strongly connected}{
      Let $T\in\mathcal{P}$ be the trail containing $e$.

      \eIf{some orientation of $T$ leaves $G$ strongly connected}{
        Apply such an orientation of $T$ to $G$
      }{
        \Return{false}
      }
    }{
      Let $T\in\mathcal{P}$ be arbitrary.

      Update $G$ by orienting $T$ in an arbitrary direction.
    }
    Remove $T$ from $\mathcal{P}$.
  }
  \Return{true}
  \caption{\label{alg:mixed}Algorithm for mixed graphs.}
\end{algorithm}

Theorem \ref{thm:mixed} gives a sufficient condition for when a strong orientation exists and we deal with the other cases by dealing with the forced edges first. However, the generalised Robbins' Theorem provides a simple equivalent condition, which we lack. Finding such an equivalent condition when you have trail decomposition is an essential open problem for strong graph orientations. Due to figure \ref{fig:counter example} in this setting one has to take into account the structure of the trail partition.

\section{Linear time algorithm}\label{sec:lin}
In this section we provide our linear time algorithm for solving the trail orientation problem in undirected graphs. For this, we make two crucial observations.  First, we show that there is an easy linear time reduction from general graphs or multigraphs to cubic multigraphs.  Second, we show that in a cubic multigraph with $n$ vertices, we can in linear time find and delete a set of edges that are at the end of their trails, such that the resulting graph has $\Omega(n)$ $3$-edge connected components.  We further show that we can compute the required orientation recursively from an orientation of each $3$-edge connected component together with the cactus of $3$-edge connected components.
Since the average size of these components is constant, we can compute the orientations of most of them in linear time.  The rest contains at most a constant fraction of the vertices, and so a simple geometric sum argment tells us that the total time is also linear.

We start out by making the following reduction.

\begin{lemma}\label{lem:cubify}
  The one-way trail problem on a $2$-edge connected graph or multigraph with $n$ vertices and $m$ edges, reduces in $\OO(m+n)$ time to the same problem on a $2$-edge connected cubic multigraph with $2m$ vertices and $3m$ edges.
\end{lemma}
\begin{proof}

  Order the edges adjacent to each vertex such that two edges that are adjacent on the same trail are consecutive in the order.  Replace each single vertex $v$ with a cycle of length $\deg(v)$, with each vertex of the new cycle inheriting a corresponding neighbour of $v$.  Note that for a vertex of degree $2$, this creates a pair of parallel edges, so the result may be a multigraph.  Since edges on the same trail are neighbours, we can make the cycle-edge between the two corresponding vertices belong to the same trail.  The rest of the cycle edges form new length $1$ trails.  This graph has exactly $2m$ vertices and $3m$ edges, and  any one-way trail orientation on this graph translates to a one-way trail orientation of the original graph.  The graph is constructed in $\OO(m+n)$ time.
  
    \begin{figure}[h!]
\centering
\begin{tikzpicture}[x=0.5cm,y=0.5cm,scale=0.75]
  \begin{scope}[
      every path/.style={
              },
      every node/.style={
                text=black,
        inner sep=1pt,
      },
            vertex set/.style={
        dashed,
      },
            vertex/.style={
        draw,
        circle,
        fill=white,
        minimum size=2mm,
        inner sep=0pt,
        outer sep=0pt,
      },
            edge/.style={blue,thick},
      undirected edge/.style={edge},
      directed edge/.style={edge,->,>=stealth'},
    ]
    \begin{scope}[      every node/.style={
        circle,
        draw,
        fill=white,
        minimum size=0.75cm,
      },
      every path/.style={
        ultra thick,
      }
    ]
      \node[label={0:$v$}] (e1) at ($(0,0)$) {};
      \node (a5) at ($(0,5)$) {};
      \node (a1) at ($(-5,-5)$) {};
      \node (a2) at ($(-5,5)$) {};
      \node (a3) at ($(5,-5)$) {};
      \node (a4) at ($(5,5)$) {};

      \draw[red] (a1) -- (e1);
      \draw[red] (a2) -- (e1);
      \draw[green] (a5) -- (e1);
      \draw[undirected edge] (a3) -- (e1);
      \draw[undirected edge] (a4) -- (e1) ;
    \end{scope}

    \begin{scope}[
                shift={(24,0)},
        every node/.style={
        circle,
        draw,
        fill=white,
        minimum size=0.75cm,
      },
      every path/.style={
        ultra thick,
      }
      ]
        \node (e5) at ($(0,4)$) {};
      \node (e1) at ($(1.5,-1.5)$) {};
      \node (e2) at ($(-1.5,1.5)$) {};
      \node (e3) at ($(1.5,1.5)$) {};
      \node (e4) at ($(-1.5,-1.5)$) {};
\node (a5) at ($(0,7)$) {};
      \node (a1) at ($(-5,-5)$) {};
      \node (a2) at ($(-5,5)$) {};
      \node (a3) at ($(5,-5)$) {};
      \node (a4) at ($(5,5)$) {};

      \draw[red] (a1) -- (e4);
      \draw[red] (a2) -- (e2);
      \draw[undirected edge] (a3) -- (e1);
      \draw[undirected edge] (a4) -- (e3);
      \draw[undirected edge] (e3) -- (e1);
      \draw[red] (e2) -- (e4);
      \draw[black] (e2) -- (e5);
      \draw[black] (e3) -- (e5);
      \draw[black] (e1) -- (e4);
      \draw[green] (e5) -- (a5);
    \end{scope}

  \end{scope}
\end{tikzpicture}
\caption{A node of degree 5 turns into a cycle of length 5 }
\label{fig:Cubic}
\end{figure}
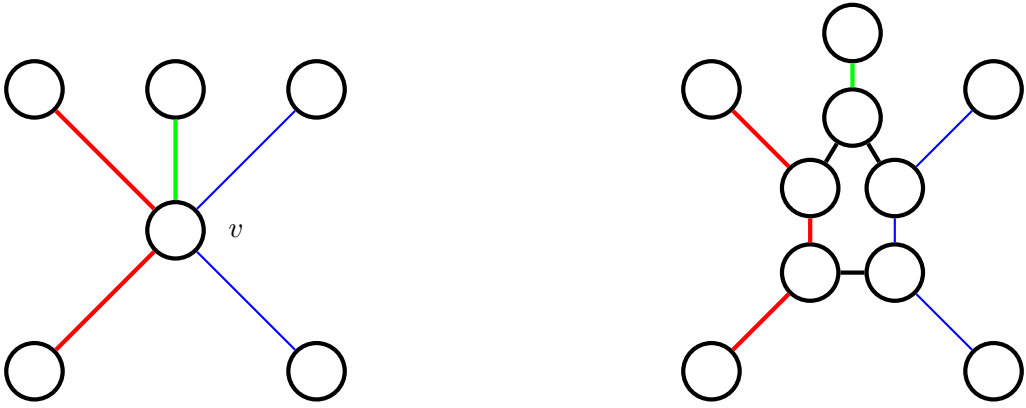
\end{proof}
Recall now that a graph $C$ is called a \emph{cactus} if it is connected and each edge is contained in at most one cycle. If $G$ is any connected graph we let $C_1,\dots,C_k$ be its $3$-edge connected components. It is well known that if we contract each of these we obtain a cactus graph. For a proof of this result see section 2.3.5 of \cite{Nagamochi:2008:AAG:1434866}. As the cuts in a contracted graph are also cuts in the original graph we have that if $G$ is $2$-edge connected then the cactus graph is $2$-edge connected. The edges of the cactus are exactly the edges of $G$ which are part of a $2$-edge cut. We will call these edges \emph{2-edge critical}.

It is easy to check that if a cactus has $m$ edges and $n$ vertices then $m\leq 2(n-1)$. We will be using this result in the proof of the following lemma.
\begin{lemma}\label{lem:manycomponents}
  Let $G=(V,E)$ be a cubic $2$-edge connected multigraph, let $X\subseteq E$, and let $F\supseteq E\setminus X$ be minimal such that $H=(V,F)$ is  $2$-edge connected.  Then $H$ has at least $\frac{2}{5}\sizeof{X}$ distinct $3$-edge connected components.
\end{lemma}
\begin{proof}
  \newcommand{\Xdel}{X_{\text{del}}}
  \newcommand{\Xkeep}{X_{\text{keep}}}
  Let $\Xdel\subseteq X$, be the set of edges deleted from $G$ to obtain $H$, and let $\Xkeep=X\setminus\Xdel$ be the remaining edges in $X$.

If $\sizeof{\Xkeep}\geq\frac{4}{5}\sizeof{X}$, then by minimality of $H$ there are at least $\sizeof{\Xkeep}$ 2-edge-critical edges in $H$ i.e. edges of the corresponding cactus, and thus at least $\frac{1}{2}\sizeof{\Xkeep}+1\geq\frac{2}{5}\sizeof{X}+1$ distinct $3$-edge connected components.

If $\sizeof{\Xkeep}\leq\frac{4}{5}\sizeof{X}$ then $\sizeof{\Xdel}\geq\frac{1}{5}\sizeof{X}$, and since $G$ is cubic and the removal of each edge creates two vertices of degree $2$ we must have that $H$ has at least $2\sizeof{\Xdel}\geq\frac{2}{5}\sizeof{X}$  distinct $3$-edge connected components.
\end{proof}

\begin{lemma}\label{lem:innertree}
  Let $G=(V,E)$ be a connected cubic multigraph with $E$ partitioned into trails.  Then $G$ has a spanning tree that contains all edges that are not at the end of their trail.
\end{lemma}
\begin{proof}
  Let $F$ be the set of edges that are not at the end of their trail.
  Since $G$ is cubic, the graph $(V,F)$ is a collection of vertex-disjoint paths, and in particular it is acyclic. Since $G$ is connected $F$ can be extended to a spanning tree. 
\end{proof}
Note that we can find this spanning tree in linear time. Indeed, we may assign weight $0$ to edges in $F$ and $1$ to the remaining edges and use the so-called\footnote{Originally discovered by Jarn\'ik~\cite{jarnik1930}, later rediscovered by Prim~\cite{prim1957}} Prim's minimal spanning tree algorithm with a suitable priority queue to find the tree.

\begin{lemma}\label{lem:constantfraction}
  Let $G=(V,E)$ be a cubic $2$-edge connected multigraph with $E$ partitioned into trails.
  Let $T$ be a spanning tree of $G$ containing all edges that are not at the end of their trail.
  Let $H$ be a minimal subgraph of $G$ that contains $T$ and is $2$-edge connected.
  Then for any $k\geq5$, less than $\frac{4}{5}\frac{k}{k-1}\sizeof{V}$ of the vertices in $H$ are in a $3$-edge connected component with at least $k$ vertices.
\end{lemma}
\begin{proof}
  Let $X$ be the set of edges that are not in $T$.  Since $G$ is cubic, $\sizeof{X}=\frac{1}{2}\sizeof{V}+1$.  By Lemma~\ref{lem:manycomponents} $H$ has at least $\frac{2}{5}\sizeof{X}>\frac{1}{5}\sizeof{V}$ $3$-edge connected components.  Each such component contains at least one vertex, so the total number of vertices in components of size at least $k$ is less than $\frac{k}{k-1}\left(\sizeof{V}-\frac{1}{5}\sizeof{V}\right)=\frac{4}{5}\frac{k}{k-1}\sizeof{V}$.
\end{proof}

\begin{definition}\label{def:componentgraph}
  Let $C$ be a $3$-edge connected component in some graph $H$, whose edges is partitioned into trails.  Define $\Gamma_H(C)$ to be the $3$-edge connected graph obtained by replacing each min-cut $\set{e,f}$ where $e=(e_1,e_2)$ and $f=(f_1,f_2)$ and $e_1,f_1\in C$ with a single new edge $(e_1,f_1)$.  Define the corresponding partition of the edges of $\Gamma_H(C)$ into trails by taking every trail that is completely contained in $C$, together with new trails combined from the fragments of the trails that were broken by the min-cuts together with the  new edges that replaced them. See figure \ref{fig:cactus}.
\end{definition}
At this point the idea of the algorithm can be explained. We remove as many of the edges, at the end of their trails, as we can still maintaining that the graph is $2$-edge connected. Lemma \ref{lem:constantfraction} guarantees that we obtain a graph $H$ with $\Omega(|V|)$ many $3$-edge connected components of size $O(1)$. We solve the problem for each $\Gamma_H(C)$ for every $3$-edge connected component. Finally, we combine the solutions for the different components like in the proof of theorem \ref{thm:exist}. 

  \begin{figure}[h!]
\centering
\begin{tikzpicture}[x=0.5cm,y=0.5cm,scale=0.4]
  \begin{scope}[
      every path/.style={
              },
      every node/.style={
                text=black,
        inner sep=1pt,
      },
            vertex set/.style={
        dashed,
      },
            vertex/.style={
        draw,
        circle,
        fill=white,
        minimum size=2mm,
        inner sep=0pt,
        outer sep=0pt,
      },
            edge/.style={blue,thick},
      undirected edge/.style={edge},
      directed edge/.style={edge,->,>=stealth'},
    ]
    \begin{scope}[      every node/.style={
        circle,
        draw,
        fill=white,
        minimum size=1.5cm,
      },
      every path/.style={
        ultra thick,
      },
      vertex set/.style={
        dashed,
      },
    ]
      \node[vertex set,label={0:$ $}] (e1) at ($(0,0)$) {$C$};
      \node[vertex set] (a2) at ($(0,10)$) {};
      \node[vertex set] (a4) at ($(-10,-10)$) {};
      \node[vertex set] (a1) at ($(-10,10)$) {};
      \node[vertex set] (a5) at ($(10,-10)$) {};
      \node[vertex set] (a3) at ($(10,10)$) {};
      \node[vertex set] (a6) at ($(10,0)$) {};

      \draw[red] (a1) -- (e1);
      \draw[red] (a1) -- (a2);
      \draw[red] (a3) -- (a2);
      \draw[red] (a3) -- (e1);
      \draw[undirected edge] (a4) -- (a5);
      \draw[undirected edge] (a4) -- (e1);
      \draw[undirected edge] (a5) -- (e1);
      \draw[green] (a6) .. controls ($(5,2)$)
      .. (e1);
      \draw[green] (a6) .. controls ($(5,-2)$)
      .. (e1);
    \end{scope}
    \begin{scope}[
                shift={(24,0)},
        every node/.style={
        circle,
        draw,
        fill=white,
        minimum size=1.5cm,
      },
      every path/.style={
        ultra thick,
      }
      ]
      \node[vertex set,label={0:$ $}] (e1) at ($(0,0)$) {$\Gamma_H(C)$};
      \draw[green] (e1) .. controls ($(5,-2)$) and ($(5,2)$) 
      .. (e1);
      \draw[red] (e1) .. controls ($(-3,5)$) and ($(3,5)$) 
      .. (e1);
      \draw[blue] (e1) .. controls ($(-3,-5)$) and ($(3,-5)$) 
      .. (e1);
    \end{scope}

  \end{scope}
\end{tikzpicture}
\caption{$3$-edge connected components, notice how every edge out from the centre is part of a cycle. This right hand shows $\Gamma_H(C)$ where $C$ is the component in the middle. }
\label{fig:cactus}
\end{figure}
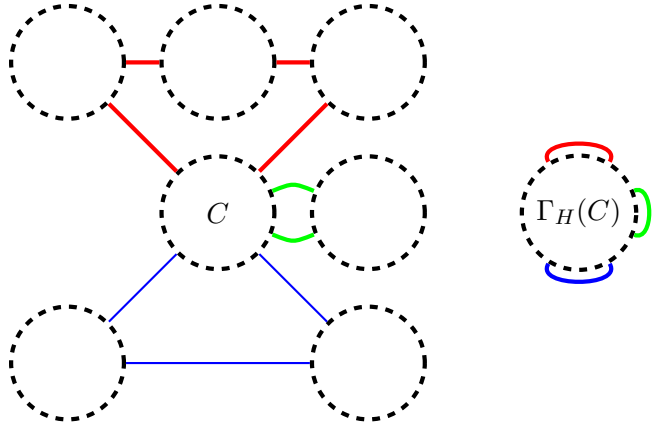

\begin{theorem}\label{thm:linear}
  The one-way trail orientation problem can be solved in $\OO(m+n)$ time on any $2$-edge connected graph or multigraph with $n$ vertices and $m$ edges.
\end{theorem}
\begin{proof}
  By Lemma~\ref{lem:cubify}, we can assume the graph is cubic. For the algorithm we will use two subroutines. First of all when we have found the minimum spanning tree $T$ containing the edges that are not on the end of their trail we can use the algorithm of Kelsen \cite{KelsenR91} to, in linear time, find a minimal (w.r.t. inclusion) subgraph $H$ of $G$ that contains $T$ and is $2$-edge connected. Secondly we will use the algorithm of Melhorn \cite{Mehlhorn2017} to, in linear time, build the cactus graph of $3$-edge connected components.
  The algorithm runs as follows:
  \begin{enumerate}
  \item Construct a spanning tree $T$ of $G$ that contains all edges that are not at the end of their trail.
  \item Construct a minimal subgraph $H$ of $G$ that contains $T$ and is $2$-edge connected\footnote{See Kelsen \cite{KelsenR91}}.
  \item Find the cactus of $3$-edge connected components of\footnote{See Melhorn \cite{Mehlhorn2017}} $H$.
  \item For each $3$-edge connected component $C_i$, construct $\Gamma_H(C_i)$.
  \item Recursively compute an orientation for each\footnote{Note that $\Gamma_H(C_i)$ is cubic unless it consists of exactly one node. In this case however we don't need to do anything.} $\Gamma_H(C_i)$.
  \item Combine the orientations from each component.
  \end{enumerate}
  First we will show correctness and then we will determine the running time.
  
  Recall that we can flip the orientation in each $\Gamma_H(C_i)$ and still obtain a strongly connected graph respecting the trails in $\Gamma_H(C_i)$. The way we construct the orientation of the edges of $G$ is by flipping the orientation of each $\Gamma_H(C_i)$ in such a way that each cycle in the cactus graph becomes a directed cycle\footnote{In practise this is done by making a DFS (or any other search tree one likes) of the cactus and repeatedly orienting each component in a way consistent with the previous ones.}. This can be done exactly because no edge of the cactus is contained in two cycles. By construction this orientation respects the trails so we need to argue that it gives a strongly connected graph.

 For showing that the resulting graph is strongly connected, first let every $3$-edge connected component be contracted, then the graph is strongly connected since the cycles of the cactus graph have become directed cycles. Now assume inductively that we have uncontracted some of the components and call this graph $G_1$. Now we uncontract another component $C$ (see figure \ref{fig:blowing up}) and obtain a new graph $G_2$ which we will show is also strongly connected. 
 If $u,v \in C$, then since $\Gamma_H(C)$ is strongly connected there is a path from $u$ to $v$ in $\Gamma_H(C)$. If this path only contains edges which are edges in $C$ clearly this path also exists in $G_2$ so we are done. If the path uses one of the added edges $(e_1,f_1)$ (without loss of generality oriented from $e_1$ to $f_1$), it is because there are edges $(e_1,e_2)$ and $(f_1,f_2)$ forming a cut and thus being part of a cycle in the cactus. In this case we use edge $(e_1,e_2)$ to leave component $C$ and then go from $e_2$ back to component $C$ which is possible since $G_1$ was strongly connected. When we get back to the component $C$ we must arrive at $f_1$ since otherwise there would be two cycles in the cactus containing the edge $(e_1,e_2)$. Hence the edge $(e_1,f_1)$ was not needed. This argument can be used for any of the edges of $\Gamma_H(C)$ that are not in $C$ and thus we can move between any two nodes in $C$. Since $G_1$ was strongly connected this suffices to show that $G_2$ is strongly connected. By induction this implies that after uncontracting all components the resulting graph is strongly connected.
 
   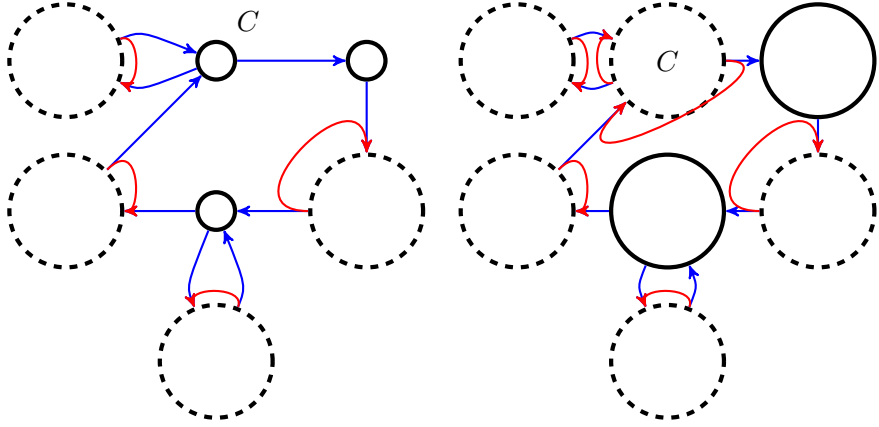
\begin{figure}[h!]
\centering
\begin{tikzpicture}[x=0.5cm,y=0.5cm,scale=0.4]
  \begin{scope}[
      every path/.style={
              },
      every node/.style={
                text=black,
        inner sep=1pt,
      },
            vertex set/.style={
        dashed,
      },
            vertex/.style={
        draw,
        circle,
        fill=white,
        minimum size=2mm,
        inner sep=0pt,
        outer sep=0pt,
      },
            edge/.style={blue,thick},
      undirected edge/.style={edge},
      directed edge/.style={edge,->,>=stealth'},
    ]
    \begin{scope}[      every node/.style={
        circle,
        draw,
        fill=white,
        minimum size=0.5cm,
      },
      every path/.style={
        ultra thick,
      },
      vertex set/.style={
        dashed,
      },
    ]
      \node[vertex set,minimum size=1.5cm] (a1) at ($(-10,10)$) {};
      \node[label={60:$C$}] (a2) at ($(0,10)$) {};
      \node (a3) at ($(10,10)$) {};
      \node[vertex set,minimum size=1.5cm] (a4) at ($(-10,0)$) {};
      \node (a5) at ($(0,0)$) {};
      \node[vertex set,minimum size=1.5cm] (a6) at ($(10,0)$) {};
      \node[vertex set,minimum size=1.5cm] (a7) at ($(0,-10)$) {};

      \draw[directed edge] (a2) -- (a3);
      \draw[directed edge] (a5) -- (a4);
      \draw[directed edge] (a4) -- (a2);
      \draw[directed edge] (a3) -- (a6);
      \draw[directed edge] (a6) -- (a5);
      \draw[directed edge] (a2) .. controls ($(-5,8)$)
      .. (a1);
      \draw[directed edge] (a1) .. controls ($(-5,12)$)
      .. (a2);
      \draw[directed edge] (a5) .. controls ($(-2,-5)$)
      .. (a7);
      \draw[directed edge] (a7) .. controls ($(2,-5)$)
      .. (a5);
      \draw[directed edge,red] (a1) .. controls ($(-5,12)$) and ($(-5,8)$)
      .. (a1);
      \draw[directed edge,red] (a4) .. controls ($(-5,5)$) and ($(-5,0)$)
      .. (a4);
      \draw[directed edge,red] (a6) .. controls ($(0,0)$) and ($(10,10)$)
      .. (a6);
      \draw[directed edge,red] (a7) .. controls ($(2,-5)$) and ($(-2,-5)$)
      .. (a7);
    \end{scope}
    \begin{scope}[
                shift={(30,0)},
        every node/.style={
        circle,
        draw,
        fill=white,
        minimum size=1.5cm,
      },
      every path/.style={
        ultra thick,
      }
      ]
      \node[vertex set,minimum size=1.5cm] (a1) at ($(-10,10)$) {};
      \node[vertex set,minimum size=1.5cm] (a2) at ($(0,10)$) {$C$};
      \node (a3) at ($(10,10)$) {};
      \node[vertex set,minimum size=1.5cm] (a4) at ($(-10,0)$) {};
      \node (a5) at ($(0,0)$) {};
      \node[vertex set,minimum size=1.5cm] (a6) at ($(10,0)$) {};
      \node[vertex set,minimum size=1.5cm] (a7) at ($(0,-10)$) {};

      \draw[directed edge] (a2) -- (a3);
      \draw[directed edge] (a5) -- (a4);
      \draw[directed edge] (a4) -- (a2);
      \draw[directed edge] (a3) -- (a6);
      \draw[directed edge] (a6) -- (a5);
      \draw[directed edge] (a2) .. controls ($(-5,8)$)
      .. (a1);
      \draw[directed edge] (a1) .. controls ($(-5,12)$)
      .. (a2);
      \draw[directed edge] (a5) .. controls ($(-2,-5)$)
      .. (a7);
      \draw[directed edge] (a7) .. controls ($(2,-5)$)
      .. (a5);
      \draw[directed edge,red] (a1) .. controls ($(-5,12)$) and ($(-5,8)$)
      .. (a1);
      \draw[directed edge,red] (a4) .. controls ($(-5,5)$) and ($(-5,0)$)
      .. (a4);
      \draw[directed edge,red] (a6) .. controls ($(0,0)$) and ($(10,10)$)
      .. (a6);
      \draw[directed edge,red] (a7) .. controls ($(2,-5)$) and ($(-2,-5)$)
      .. (a7);
      \draw[directed edge,red] (a2) .. controls ($(-5,8)$) and ($(-5,12)$)
      .. (a2);
      \draw[directed edge,red] (a2) .. controls ($(10,10)$) and ($(-10,0)$)
      .. (a2);
    \end{scope}

  \end{scope}
\end{tikzpicture}
\caption{Before and after uncontracting component $C$}
\label{fig:blowing up}
\end{figure}

  Now for the running time. By Lemma~\ref{lem:constantfraction} each level of recursion reduces the number of vertices in ``large'' components by a constant fraction, for instance for $k=10$ we reduce the number of vertices in large components by a factor of $\frac{1}{9}$. Let $f(n)$ be the worst case running time with $n$ nodes for a cubic graph, and pick $c$ large enough such that $cn$ is larger than the time it takes to go through steps $1$-$4$ and $6$ as well as computing the orientations in the ``small'' components. Let $a_1,\dots,a_k$ be the number of vertices in the ``large'' $3$-edge connected components. Then $\sum_i a_i \leq \frac{8n}{9}$ and
\begin{align*}
    f(n) \leq cn+ \sum_i f(a_i) \\
\end{align*}
Inductively we may assume that $f(a_i)\leq 9cn$ and thus obtain
\begin{align*}
    f(n) \leq cn+ \sum_i f(a_i)\leq cn+\sum_i 9ca_i = cn+8cn=9cn
\end{align*}
proving that $f(n)\leq 9cn$ for all $n$.

\end{proof}

\begin{algorithm}[h!]
     \KwIn{An undirected multigraph $G$ and a partition $\mathcal{P}$ of the edges of $G$ into trails.}
  \KwOut{True if $(G,\mathcal{P})$ has a strong trail orientation, otherwise false.  $G$ is modified in place, either to have such a strong trail orientation, or to a forced graph that is not strongly connected.}
  Construct a spanning tree $T$ of $G$ that contains all edges that are not at the end of their trail.

  Construct a minimal subgraph $H$ of $G$ that contains $T$ and is $2$-edge connected.

  Find the cactus $C$ of $3$-edge connected components of $H$.

  \For{each $3$-edge connected component $C_i$ in $C$ in DFS preorder}{
    Construct $G_i=\Gamma_H(C_i)$.

    Recursively compute an orientation for $G_i$.

    \If{the orientation of $G_i$ is not compatible with its DFS parent}{
      Flip orientation of $G_i$
    }
  }

  \For{each edge $e$ deleted from $G$ to create $H$}{
    \eIf{no edge on the trail of $e$ has been oriented yet}{
      Pick an arbitrary orientation for $e$.
    }{
      Set the orientation of $e$ to follow the trail.
    }
  }
  \caption{\label{alg:linear}Linear time algorithm for cubic graphs.}
\end{algorithm}
\section{Open problems}
We here mention two problems concerning trail orientations which remain open. 

First of all, our linear time algorithm for finding trail orientations only works for undirected graphs and it doesn't seem to generalise to the trail orientation problem for mixed graphs. It would be interesting to know whether there also exists a linear time algorithm working for mixed graphs. If so it would complete the picture of how fast an algorithm we can obtain for any variant of the trail orientation problem.

Secondly, our sufficient condition for when it is possible to solve the trail orientation problem for mixed multigraphs is clearly not necessary. It would be interesting to know whether there is a simple necessary and sufficient condition like there is in the undirected case. Since in the mixed case the answer to the problem actually depends on the given trail decomposition and not just on the connectivity of the graph it is harder to provide such a condition. One can give the following condition. It is possible to orient the trails making the resulting graph strongly connected if and only if when we repeatedly direct the forced trails end up with a graph satisfying our condition in theorem \ref{thm:mixed}. This condition is not simple and is not easy to check directly. Is there a more natural condition?
\newpage
\bibliographystyle{plain}
\bibliography{bibl}
\end{document}